\documentclass[letterpaper, 10 pt, conference]{ieeeconf}

\usepackage{amsfonts,amsmath,amssymb} 

\DeclareMathOperator{\diag}{diag}         
\def\rbb{\mathbb{R}}
\def\trp{^T}
\def\diag{{\rm diag}}

\def\half{\frac{1}{2}}
\usepackage{psfrag,color}
\usepackage{enumerate,cite,latexsym,graphicx,subfig}
\newtheorem{theorem}{Theorem}
\newtheorem{lemma}{Lemma}

\newtheorem{definition}{Definition}

\newtheorem{remark}{Remark}

\title{\LARGE \bf Time Averaged Consensus in a Direct Coupled Distributed Coherent Quantum Observer}

\author{Ian R.~Petersen %
\thanks{This work was supported by the
Australian Research Council (ARC) and the Air Force Office of Scientific
Research (AFOSR). This material is based on research sponsored by the
Air Force Research Laboratory, under agreement number FA2386-12-1-4075.  The U.S. Government is authorized to reproduce and
distribute reprints for Governmental purposes notwithstanding any
copyright notation thereon.
The views and conclusions contained herein are those of the authors
and should not be interpreted as necessarily representing the official
policies or endorsements, either expressed or implied, of the Air
Force Research Laboratory or the U.S. Government. }%
\thanks{Ian R. Petersen is with the School of  Engineering and Information Technology, 
        University of New South Wales at the Australian Defence Force Academy, Canberra ACT 2600, Australia.
         {\tt\small i.r.petersen@gmail.com} } 
}%

\begin{document}

\maketitle
\thispagestyle{empty}
\pagestyle{empty}

\begin{abstract}
This paper considers the problem of constructing a distributed direct coupling quantum observer for a closed linear quantum system. The proposed distributed observer consists of a network of quantum harmonic oscillators and it is shown that the distributed observer converges to a consensus in a time averaged sense in which each component of the observer estimates the specified output of the quantum plant.  An example and simulations are included to illustrate the properties of the distributed observer.
\end{abstract}

\section{Introduction} \label{sec:intro}
In recent years, there has been significant interest in controlling networks of multi-agent systems to achieve a consensus among the agents; e.g., see \cite{LZHD14,SJ13,ME10,XB05,JLM03}. In particular, some authors have looked at the problem of consensus in distributed estimation problems; e.g., see \cite{OS05,OS09}. Furthermore, issues of consensus have been considered in networked quantum systems; see \cite{SSR10,MST13,MTS13,TMS14,SDPJ1a}. This work is motivated by the fact that it is becoming increasingly possible for quantum control experiments to involve the networked interconnection of many quantum elements and these quantum networks will have important applications in problems such as quantum communication and quantum information processing. Also, many macroscopic systems can be regarded as consisting of a large quantum network. 

In this paper, we build on the papers \cite{PET14Aa,PET14Ba} which considered the problem of constructing a direct coupling quantum observer for a given quantum system. The problem of constructing an observer for a linear quantum system has been considered in a number of recent papers; e.g, see \cite{MJ12a,VP9a}. The control of linear quantum systems has been of considerable interest in recent years; e.g., see \cite{JNP1,NJP1,ShP5}.
Such linear quantum systems commonly arise in the area of quantum optics; e.g., see
\cite{GZ00,BR04}. For such  system models, an important class of  control problems are  coherent
quantum feedback control problems; e.g., see \cite{JNP1,NJP1,MaP3,MAB08,ZJ11,VP4,VP5a,HM12}. In these  control problems, both the plant and the controller are quantum systems and the controller is designed to optimize some performance index. The coherent quantum observer problem can be regarded as a special case of the coherent
quantum feedback control problem in which the objective of the observer is to estimate the system variables of the quantum plant. The papers \cite{PET14Aa,PET14Ba} considered a direct coupling coherent observer problem in which the observer is directly coupled to the plant and not coupled via a field as in previous papers. This leads the papers \cite{PET14Aa,PET14Ba} to consider a notion of time-averaged convergence for the observers. 

In this paper, we extend the results of \cite{PET14Aa} to consider a direct coupled distributed quantum observer which is constructed via the direct connection of many quantum harmonic oscillators. We show that this quantum network can be constructed so that each output of the direct coupled distributed quantum observer converges to the plant output of interest in a time averaged sense. This is a form of time averaged quantum consensus for the quantum networks under consideration.

\section{Quantum  Systems}
In the distributed quantum observer problem under consideration, both the quantum plant and the distributed quantum observer are linear quantum systems; see also \cite{JNP1,GJ09,ZJ11}. We will restrict attention to closed linear quantum systems which do not interact with an external environment. 
The quantum mechanical behavior of a linear quantum system is described in terms of the system \emph{observables} which are self-adjoint operators on an underlying infinite dimensional complex Hilbert space $\mathfrak{H}$.   The commutator of two scalar operators $x$ and $y$ on ${\mathfrak{H}}$ is  defined as $[x, y] = xy - yx$.~Also, for a  vector of operators $x$ on ${\mathfrak H}$, the commutator of ${x}$ and a scalar operator $y$ on ${\mathfrak{H}}$ is the  vector of operators $[{x},y] = {x} y - y {x}$, and the commutator of ${x}$ and its adjoint ${x}^\dagger$ is the  matrix of operators 
\[ [{x},{x}^\dagger] \triangleq {x} {x}^\dagger - ({x}^\# {x}^T)^T, \]
where ${x}^\# \triangleq (x_1^\ast\; x_2^\ast \;\cdots\; x_n^\ast )^T$ and $^\ast$ denotes the operator adjoint. 

The dynamics of the closed linear quantum systems under consideration are described by non-commutative differential equations of the form
\begin{eqnarray}
\dot x(t) &=& Ax(t); \quad x(0)=x_0
 \label{quantum_system}
\end{eqnarray}
where $A$ is a real matrix in $\rbb^{n
\times n}$, and $ x(t) = [\begin{array}{ccc} x_1(t) & \ldots &
x_n(t)
\end{array}]\trp$ is a vector of system observables; e.g., see \cite{JNP1}. Here $n$ is assumed to be an even number and $\frac{n}{2}$ is the number of modes in the quantum system. 

The initial system variables $x(0)=x_0$ 
are assumed to satisfy the {\em commutation relations}
\begin{equation}
[x_j(0), x_k(0) ] = 2 i \Theta_{jk}, \ \ j,k = 1, \ldots, n,
\label{x-ccr}
\end{equation}
where $\Theta$ is a real skew-symmetric matrix with components
$\Theta_{jk}$.  In the case of a
single quantum harmonic oscillator, we will choose $x=(x_1, x_2)^T$ where
$x_1=q$ is the position operator, and $x_2=p$ is the momentum
operator.  The
commutation relations are  $[q,p]=2 i$.
In general, the matrix $\Theta$ is assumed to be  of the  form
\begin{equation}
\label{Theta}
\Theta=\diag(J,J,\ldots,J)
\end{equation}
 where $J$ denotes the real skew-symmetric $2\times 2$ matrix
$$
J= \left[ \begin{array}{cc} 0 & 1 \\ -1 & 0
\end{array} \right].$$

The system dynamics (\ref{quantum_system}) are determined by the system Hamiltonian which is a 
which is a self-adjoint operator on the underlying  Hilbert space  $\mathfrak{H}$. For the linear quantum systems under consideration, the system Hamiltonian will be a
quadratic form
$\mathcal{H}=\half x(0)\trp R x(0)$, where $R$ is a real symmetric matrix. Then, the corresponding matrix $A$ in 
(\ref{quantum_system}) is given by 
\begin{equation}
A=2\Theta R \label{eq_coef_cond_A}.
\end{equation}
 where $\Theta$ is defined as in (\ref{Theta}).
e.g., see \cite{JNP1}.
In this case, the  system variables $x(t)$ 
will satisfy the {\em commutation relations} at all times:
\begin{equation}
\label{CCR}
[x(t),x(t)^T]=2{\pmb i}\Theta \ \mbox{for all } t\geq 0.
\end{equation}
That is, the system will be \emph{physically realizable}; e.g., see \cite{JNP1}.

\begin{remark}
\label{R1}
Note that that the Hamiltonian $\mathcal{H}$ is preserved in time for the system (\ref{quantum_system}). Indeed,
$ \mathcal{\dot H} = \frac{1}{2}\dot{x}^TRx+\frac{1}{2}x^TR\dot{x} = -x^TR\Theta R x + x^TR\Theta R x = 0$ since $R$ is symmetric and $\Theta$ is skew-symmetric.
\end{remark}

\section{Direct Coupling Distributed Coherent Quantum Observers}
In our proposed direct coupling coherent quantum observer, the quantum plant is a single quantum harmonic oscillator which is a linear quantum system of the form (\ref{quantum_system}) described by the non-commutative differential equation
\begin{eqnarray}
\dot x_p(t) 
&=& A_px_p(t); \quad x_p(0)=x_{0p}; \nonumber \\
z_p(t) &=& C_px_p(t)
 \label{plant}
\end{eqnarray}
where $z_p(t)$ denotes the vector of system variables to be estimated by the observer and $ A_p \in \rbb^{2
\times 2}$, $C_p\in \rbb^{1 \times 2}$. 
It is assumed that this quantum plant corresponds to a plant Hamiltonian
$\mathcal{H}_p=\half x_p(0)\trp R_p x_p(0)$. Here $x_p = \left[\begin{array}{l}q_p\\p_p\end{array}\right]$ where
$q_p$ is the plant position operator and $p_p$ is the plant momentum operator. 

We now describe the linear quantum system of the form (\ref{quantum_system}) which will correspond to the distributed quantum observer; see also \cite{JNP1,GJ09,ZJ11}. 
This system is described by a non-commutative differential equation of the form
\begin{eqnarray}
\dot x_o(t) &=& A_ox_o(t);\quad x_o(0)=x_{0o};\nonumber \\
z_o(t) &=& C_ox_o(t)
 \label{observer}
\end{eqnarray}
where the observer output $z_o(t)$ is the distributed observer estimate vector and $ A_p \in \rbb^{n_o
\times n_o}$, $C_o\in \rbb^{\frac{n_o}{2} \times n_o}$.  Also,  $x_o(t)$  is a vector of self-adjoint 
non-commutative system variables; e.g., see \cite{JNP1}. We assume the distributed observer order $n_o$  is an even number with $N=\frac{n_o}{2}$ being the number of elements in the distributed quantum observer. We also assume that the plant variables commute with the observer variables. The system dynamics (\ref{observer}) are determined by the observer system Hamiltonian which is a 
which is a self-adjoint operator on the underlying  Hilbert space for the observer. For the distributed quantum observer under consideration, this Hamiltonian is given by a 
quadratic form:
$\mathcal{H}_o=\half x_o(0)\trp R_o x_o(0)$, where $R_o$ is a real symmetric matrix. Then, the corresponding matrix $A_o$ in 
(\ref{observer}) is given by 
\begin{equation}
A_o=2\Theta R_o \label{eq_coef_cond_Ao}
\end{equation}
 where $\Theta$ is defined as in (\ref{Theta}). Furthermore, we will assume that the distributed quantum observer has a chain structure and is coupled to the quantum plant as shown in Figure \ref{F3}. 
\begin{figure}[htbp]
\begin{center}
\includegraphics[width=8cm]{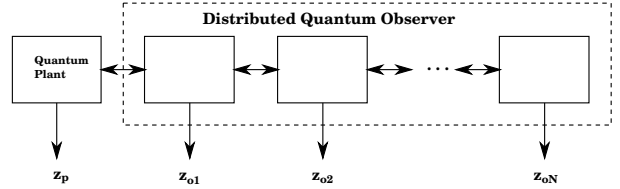}
\end{center}
\caption{Distributed Quantum Observer.}
\label{F3}
\end{figure}
This corresponds to an observer Hamiltonian of the form
\begin{eqnarray*}
\mathcal{H}_o&=&\half x_o(0)\trp R_o x_o(0) \nonumber \\
&=& \half\sum_{i=1}^{N}x_{oi}(0)\trp R_{oi} x_{oi}(0)\nonumber \\
&&+\sum_{i=1}^{N-1}x_{oi}(0)\trp R_{c(i+1)} x_{o(i+1)}(0)
\end{eqnarray*}
where the vector of observer system variables $x_o$ is partitioned according to each element of the distributed quantum observer as follows
\[
x_o = \left[\begin{array}{l}x_{o1}\\x_{o2}\\\vdots\\x_{oN}\end{array}\right].
\]
We assume that the variables for each element of the distributed quantum observer commute with the variables of all other elements of the distributed quantum observer; i.e., 
\[
[x_{oi},x_{oj}^T]=0 ~~ \forall ~~i\neq j.
\]
Here, $x_{oi} = \left[\begin{array}{l}q_{oi}\\p_{oi}\end{array}\right]$ for $i=1,2,\ldots,N$ where
$q_{oi}$ is the  position operator for the $i$th observer element  and $p_{oi}$ is the  momentum operator for the $i$th observer element. 

In addition, we define a coupling Hamiltonian which defines the coupling between the quantum plant and the first element of the distributed quantum observer:
\[
\mathcal{H}_c = x_{p}(0)\trp R_{c1} x_{o1}(0).
\]
Furthermore, we write
\[
z_o = \left[\begin{array}{l}z_{o1}\\z_{o2}\\\vdots\\z_{oN}\end{array}\right]
\]
where 
\[
z_{oi} = C_{oi}x_{oi} \mbox{ for }i=1,2,\ldots,N.
\]
Note that $R_{oi} \in \rbb^{2 \times 2}$, $R_{ci} \in \rbb^{2 \times 2}$, $C_{oi} \in \rbb^{1 \times 2}$, and each matrix $R_{oi}$ is symmetric for $i=1,2,\ldots,N$.

The augmented quantum linear system consisting of the quantum plant and the distributed  quantum observer is then a quantum system of the form (\ref{quantum_system}) described by the total Hamiltonian
\begin{eqnarray}
\mathcal{H}_a &=& \mathcal{H}_p+\mathcal{H}_c+\mathcal{H}_o\nonumber \\
 &=& \half x_a(0)\trp R_a x_a(0)
\label{total_hamiltonian}
\end{eqnarray}
where
\begin{equation}
\label{xaRa}
x_a = \left[\begin{array}{l}x_p\\x_{o1}\\x_{o2}\\\vdots\\x_{oN}\end{array}\right],~~
R_a = \left[\begin{array}{lllll}R_p & R_{c1} & & &\\
R_{c1}^T & R_{o1} & R_{c2} &  & 0 \\
& R_{c2}^T & R_{o2} & \ddots &\\
0 & & \ddots & \ddots  & R_{cN}\\
&&& R_{cN}^T & R_{oN}
\end{array}\right].
\end{equation}
 Then, using (\ref{eq_coef_cond_A}), it follows that the augmented quantum linear system is described by the equations
\begin{eqnarray}
\left[\begin{array}{l}\dot x_p(t)\\\dot x_{o1}(t)\\\dot x_{o2}(t)\\\vdots\\\dot x_{oN}(t)\end{array}\right] &=& 
A_a\left[\begin{array}{l} x_p(t)\\ x_{o1}(t)\\x_{o2}(t)\\\vdots\\x_{oN}(t)\end{array}\right];\nonumber \\
z_p(t) &=& C_px_p(t);\nonumber \\
z_o(t) &=& C_ox_o(t)
\label{augmented_system}
\end{eqnarray}
where $A_a = 2\Theta R_a$ and 
\[
C_o = \left[\begin{array}{llll}C_{o1} & & &\\
                               & C_{o2}& 0 &\\
                               & 0 & \ddots & \\
                               &&& C_{oN}
\end{array}\right].
\]

We now formally define the notion of a direct coupled linear quantum observer.

\begin{definition}
\label{D1}
The matrices $R_{o1}$, $R_{o2}$, $\ldots$, $R_{oN}$, $R_{c1}$, $R_{c2}$, $\ldots$, $R_{cN}$,  $C_{o1}$, $C_{o2}$, $\ldots$, $C_{oN}$ define a {\em distributed linear quantum observer} achieving time-averaged consensus convergence for the quantum linear plant (\ref{plant}) if the corresponding augmented linear quantum system (\ref{augmented_system}) is such that
\begin{equation}
\label{average_convergence}
\lim_{T \rightarrow \infty} \frac{1}{T}\int_{0}^{T}(\left[\begin{array}{l}1\\1\\\vdots\\1\end{array}\right]z_p(t) 
- z_o(t))dt = 0.
\end{equation}
\end{definition}

\begin{remark}
\label{R2}
Note that the above definition requires that the time average of each observer element output $z_{oi}(t)$ converges to the time average of the plant output $z_p(t)$. That is, an averaged consensus is reached by the observer element outputs. 
\end{remark}
\section{Constructing a Direct Coupling Distributed Coherent Quantum Observer}
We now describe the construction of a distributed linear quantum observer.  In this section, we assume that  $A_p =0$ in (\ref{plant}). This corresponds to $R_p = 0$ in the plant Hamiltonian. It follows from (\ref{plant}) that the plant system variables $x_p(t)$ will remain fixed if the plant is not coupled to the observer. However, when the plant is coupled to the quantum observer this will no longer be the case. We will show that if the distributed quantum observer is suitably designed, the plant quantity to be estimated  $z_p(t)$ will remain fixed and the condition (\ref{average_convergence}) will be satisfied. 

We assume that the matrices $R_{ci}$, $R_{oi}$ are of the form  
\begin{equation}
\label{RciRoi}
R_{ci} = \alpha_i\beta_i^T,~~R_{oi} = \omega_iI
\end{equation}
where $\alpha_i  \in \rbb^{2}$, $\beta_i \in \rbb^{2}$ and $\omega_i > 0$  for $i=1,2,\ldots,N$. Also, we assume that $\alpha_1 =  C_p^T$. 

We will show that these assumptions imply that the quantity $z_p(t) = C_px_p(t)$ will be constant for the augmented quantum system (\ref{augmented_system}). Indeed, it follows from (\ref{augmented_system}), (\ref{xaRa}), (\ref{Theta}) that
\[
\dot x_p(t) = 2JR_{c1}x_{o1}(t) = 2J\alpha_1\beta_1^Tx_{o1}(t).
\]
Hence, 
\[
\dot z_p(t) = 2C_pJ\alpha_1\beta_1^Tx_{o1}(t)= 2\alpha_1^TJ\alpha_1\beta_1^Tx_{o1}(t) = 0
\]
since $J$ skew-symmetric implies $\alpha_1^TJ\alpha_1= 0$. Therefore, 
\begin{equation}
\label{zp_const}
z_p(t) = z_p(0) = z_p
\end{equation}
for all $t\geq 0$. 

It now follows from (\ref{augmented_system}) that we can write
\begin{eqnarray*}
\left[\begin{array}{l}\dot x_{o1}(t)\\\dot x_{o2}(t)\\\vdots\\\dot x_{oN}(t)\end{array}\right] &=& 
A_o\left[\begin{array}{l}x_{o1}(t)\\x_{o2}(t)\\\vdots\\x_{oN}(t)\end{array}\right]
+2\left[\begin{array}{l}J\beta_1\alpha_1^T\\0\\\vdots\\0\end{array}\right]x_p(t);\nonumber \\
&=& A_o\left[\begin{array}{l}x_{o1}(t)\\x_{o2}(t)\\\vdots\\x_{oN}(t)\end{array}\right]
+B_oz_p
\end{eqnarray*}
where
\begin{small}
\[
A_o = 2\left[\begin{array}{lllll}\omega_1J & J\alpha_2 \beta_2^T & & &\\
J \beta_2\alpha_2^T & \omega_2J & J\alpha_3 \beta_3^T&  & 0 \\
& J \beta_3\alpha_3^T & \omega_3J & \ddots &\\
0 & & \ddots & \ddots  & J\alpha_N \beta_N^T\\
&&& J \beta_N\alpha_N^T & \omega_NJ
\end{array}\right]
\]
\end{small}
and
\[
B_o = \left[\begin{array}{l}2J\beta_1\\0\\\vdots\\0\end{array}\right].
\]

To construct a suitable distributed quantum observer, we will further assume that 
\begin{equation}
\label{alphabeta}
\alpha_i = \alpha = C_p^T, ~~\beta_i = -\mu_i \alpha \mbox{ and }C_{oi} = C_p= \alpha^T
\end{equation}
for all $i=1,2,\ldots,N$ where each $\mu_i > 0$. In order to construct suitable values for the quantities $\mu_i$ and $\omega_i$, we require that 
\begin{equation}
\label{xtilde}
A_o\left[\begin{array}{l}\alpha\\\alpha\\\vdots\\\alpha\end{array}\right]+B_o\|\alpha\|^2 = 0.
\end{equation}
This will ensure that the quantity
\begin{equation}
\label{tildexo}
\tilde x_o = x_o - \frac{1}{\|\alpha\|^2}\left[\begin{array}{l}\alpha\\\alpha\\\vdots\\\alpha\end{array}\right]z_p
\end{equation}
 will satisfy the non-commutative differential equation
\begin{equation}
\label{xtildedot}
\dot{\tilde x}_o = A_o \tilde x_o.
\end{equation}
This, combined with the fact that
\begin{eqnarray}
\label{Coxtilde}
\lefteqn{C_o\frac{1}{\|\alpha\|^2}\left[\begin{array}{l}\alpha\\\alpha\\\vdots\\\alpha\end{array}\right]z_p}\nonumber \\
&=& \frac{1}{\|\alpha\|^2}\left[\begin{array}{llll}\alpha^T & & &\\
                               & \alpha^T & 0 &\\
                               & 0 & \ddots & \\
                               &&& \alpha^T
\end{array}\right]\left[\begin{array}{l}\alpha\\\alpha\\\vdots\\\alpha\end{array}\right]z_p \nonumber \\
&=& \left[\begin{array}{l}1\\1\\\vdots\\1\end{array}\right]z_p
\end{eqnarray}
will be used in establishing the condition (\ref{average_convergence}) for the distributed quantum observer. 

Now we require
\begin{eqnarray*}
\lefteqn{A_o\left[\begin{array}{l}\alpha\\\alpha\\\vdots\\\alpha\end{array}\right]+B_o\|\alpha\|^2}\nonumber \\
&=&
2\left[\begin{array}{c}
\omega_1J\alpha-\mu_2J\alpha\|\alpha\|^2-\mu_1J\alpha\|\alpha\|^2\\
-\mu_2J \alpha\|\alpha\|^2+ \omega_2J\alpha- \mu_3J\alpha\|\alpha\|^2\\ 
-\mu_3J\alpha\|\alpha\|^2+\omega_3J\alpha-\mu_4J\alpha \|\alpha\|^2\\
\vdots\\
-\mu_NJ \alpha\|\alpha\|^2+\omega_NJ\alpha
\end{array}\right]\nonumber \\
&=& 0.
\end{eqnarray*}
This will be satisfied if and only if 
\[
\left[\begin{array}{c}
\omega_1-\mu_2\|\alpha\|^2-\mu_1\|\alpha\|^2\\
-\mu_2\|\alpha\|^2+ \omega_2- \mu_3\|\alpha\|^2\\ 
-\mu_3\|\alpha\|^2+ \omega_3-\mu_4 \|\alpha\|^2\\
\vdots\\
-\mu_N\|\alpha\|^2+\omega_N
\end{array}\right] = 0.
\]
That is, we require that
\begin{equation}
\label{mui}
\omega_i=(\mu_i+\mu_{i+1})\|\alpha\|^2
\end{equation}
for $i=1,2,\ldots,N-1$ and
\begin{equation}
\label{muN}
\omega_N=\mu_N\|\alpha\|^2. 
\end{equation}

To show that the above candidate distributed quantum observer leads to the satisfaction of the condition (\ref{average_convergence}), we first note that $\tilde x_o$ defined in (\ref{tildexo}) will satisfy (\ref{xtildedot}). If we can show that 
\begin{equation}
\label{xtildeav}
\lim_{T \rightarrow \infty} \frac{1}{T}\int_{0}^{T}\tilde x_o(t)dt = 0
\end{equation}
 then it will follow from (\ref{Coxtilde}) and (\ref{xtilde}) that (\ref{average_convergence}) is satisfied. In order to establish (\ref{xtildeav}), we first note that we can write
\[
A_o = 2\Theta R_o
\]
where
\begin{small}
\begin{eqnarray*}
\lefteqn{R_o = }\nonumber \\
&\left[\begin{array}{rrrrr}\omega_1I & -\mu_2\alpha\alpha^T & & &\\
 -\mu_2\alpha\alpha^T & \omega_2I & -\mu_3\alpha\alpha^T&  & 0 \\
&  -\mu_3\alpha\alpha^T & \omega_3I & \ddots &\\
0 & & \ddots & \ddots  & -\mu_N\alpha\alpha^T\\
&&&  -\mu_N\alpha\alpha^T & \omega_NI
\end{array}\right].&
\end{eqnarray*}
\end{small}
We will now show that the symmetric matrix $R_o$ is positive-definite.

\begin{lemma}
\label{L1}
The matrix $R_o$ is positive definite.
\end{lemma}

\begin{proof}
In order to establish this lemma, let 
\[
x_o =  \left[\begin{array}{l}x_{o1}\\x_{o2}\\\vdots\\x_{oN}\end{array}\right]
\]
be a non-zero real vector. Then
\begin{eqnarray*}
x_o^TR_ox_o &=& \omega_1\|x_{o1}\|^2-2\mu_2x_{o1}^T\alpha x_{o2}^T\alpha+\omega_2\|x_{o2}\|^2\nonumber \\
&& -2 \mu_3 x_{o2}^T\alpha x_{o3}^T\alpha+\omega_3\|x_{o3}\|^2\nonumber \\
&& \vdots \nonumber \\
&&-2 \mu_N x_{oN-1}^T\alpha x_{oN}^T\alpha+\omega_N\|x_{oN}\|^2\nonumber \\
&\geq& \omega_1\|x_{o1}\|^2-2\mu_2\|x_{o1}\|\|x_{o2}\|\|\alpha\|^2+\omega_2\|x_{o2}\|^2\nonumber \\
&& -2 \mu_3 \|x_{o2}|\|x_{o3}\|\|\alpha\|^2+\omega_3\|x_{o3}\|^2\nonumber \\
&& \vdots \nonumber \\
&&-2 \mu_N \|x_{oN-1}\|\|x_{oN}\|\|\alpha\|^2+\omega_N\|x_{oN}\|^2
\end{eqnarray*}
using the Cauchy-Schwarz inequality. 
Hence,
\begin{eqnarray}
\label{Roineq}
x_o^TR_ox_o&\geq& \omega_1\|x_{o1}\|^2-2\tilde \mu_2\|x_{o1}\|\|x_{o2}\|+\omega_2\|x_{o2}\|^2\nonumber \\
&& -2 \tilde \mu_3 \|x_{o2}|\|x_{o3}\|+\omega_3\|x_{o3}\|^2\nonumber \\
&& \vdots \nonumber \\
&&-2 \tilde \mu_N \|x_{oN-1}\|\|x_{oN}\|+\omega_N\|x_{oN}\|^2
\end{eqnarray}
where
\begin{equation}
\label{tildeomegamu}
\tilde \mu_i =\mu_i \|\alpha\|^2
\end{equation}
for $i=1,2,\ldots,N$. 
Thus, (\ref{Roineq}) implies
\[
x_o^TR_ox_o \geq \check x_o^T\tilde R_o\check x_o
\]
where
\[
\check x_o = \left[\begin{array}{c}\|x_{o1}\|\\\|x_{o2}\|\\\vdots\\\|x_{oN}\|\end{array}\right]
\]
and 
\[
\tilde R_o = 
\left[\begin{array}{rrrrr}\omega_1 & -\tilde \mu_2 & & &\\
 -\tilde \mu_2 & \omega_2 & -\tilde \mu_3&  & 0 \\
&  -\tilde \mu_3 & \omega_3 & \ddots &\\
0 & & \ddots & \ddots  & -\tilde \mu_N\\
&&&  -\tilde \mu_N & \omega_N
\end{array}\right].
\]

Now the vector $\check x_o$ will be non-zero if and only if the vector $x_o$ is non-zero. Hence, the matrix $R_o$ will be positive-definite if we can show that the matrix $\tilde R_o$ is positive-definite. In order to establish this fact, we first note that (\ref{mui}), (\ref{muN}) and (\ref{tildeomegamu}) imply that
\[
\omega_i=\tilde \mu_i+\tilde \mu_{i+1}
\]
for $i=1,2,\ldots,N-1$ and
\[
\omega_N=\tilde \mu_N.
\]
Hence, we can write
\begin{eqnarray*}
\tilde R_o &=& 
\left[\begin{array}{rrrrr}\tilde \mu_1+\tilde \mu_{2}  & -\tilde \mu_2 & & &\\
 -\tilde \mu_2 & \tilde \mu_2+\tilde \mu_{3} & -\tilde \mu_3&  & 0 \\
&  -\tilde \mu_3 & \tilde \mu_3+\tilde \mu_{4} & \ddots &\\
0 & & \ddots & \ddots  & -\tilde \mu_N\\
&&&  -\tilde \mu_N & \tilde \mu_N
\end{array}\right] \nonumber \\
&=& \tilde R_{o1} + \tilde R_{o2}
\end{eqnarray*}
where
\[
\tilde R_{o1} = \left[\begin{array}{rrrr}
\tilde \mu_1  & 0 &\ldots  & 0\\
0 & 0  &\ldots  & 0\\
\vdots & & & \vdots\\
0 & 0  &\ldots  & 0
\end{array}\right] \geq 0
\]
and
\[
\tilde R_{o2} = \left[\begin{array}{rrrrr}\tilde \mu_{2}  & -\tilde \mu_2 & & &\\
 -\tilde \mu_2 & \tilde \mu_2+\tilde \mu_{3} & -\tilde \mu_3&  & 0 \\
&  -\tilde \mu_3 & \tilde \mu_3+\tilde \mu_{4} & \ddots &\\
0 & & \ddots & \ddots  & -\tilde \mu_N\\
&&&  -\tilde \mu_N & \tilde \mu_N
\end{array}\right].
\]

Now the matrix $\tilde R_{o2}$ is the Laplacian matrix for the weighted graph shown in Figure \ref{F4}.
\begin{figure}[htbp]
\begin{center}
\includegraphics[width=4cm]{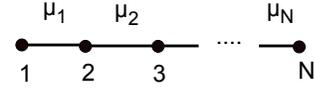}
\end{center}
\caption{Underlying weighted graph for distributed quantum observer. This corresponds to the observer structure shown in Figure \ref{F3}.}
\label{F4}
\end{figure}

Since this graph is a connected graph, it follows that the matrix $\tilde R_{o2}$ is positive-semidefinite with null space equal to 
\[
\mathcal{N}(\tilde R_{o2}) = \mbox{span}\{\left[\begin{array}{l}1\\1\\\vdots\\1\end{array}\right]\}.
\]
The fact that $\tilde R_{o1} \geq 0$ and $\tilde R_{o2} \geq 0$ implies that $\tilde R_{o} \geq 0$. In order to show that $\tilde R_{o} > 0$, suppose that $x$ is a non-zero vector in $\mathcal{N}(\tilde R_{o})$. It follows that 
\[
x^T\tilde R_{o}x = x^T\tilde R_{o1}x+x^T\tilde R_{o2}x = 0.
\]
Since $\tilde R_{o1} \geq 0$ and $\tilde R_{o2} \geq 0$, $x$ must be contained in the null space of $\tilde R_{o1}$ and the null space of $\tilde R_{o2}$. Therefore $x$ must be of the form
\[
x = \gamma \left[\begin{array}{l}1\\1\\\vdots\\1\end{array}\right]
\]
where $\gamma \neq 0$. However, then
\[
x^T\tilde R_{o1}x = \gamma^2 \tilde \mu_1 \neq 0
\]
and hence $x$ cannot be in the null space of $\tilde R_{o1}$. Thus, we can conclude that the matrix $\tilde R_{o}$ is positive definite and hence, the matrix  $R_{o}$ is positive definite. This completes the proof of the lemma. 
\end{proof}

We now verify that the condition (\ref{average_convergence}) is satisfied for the distributed  quantum observer under consideration. We recall from Remark \ref{R1} that the quantity $\half \tilde x_o(t)\trp R_o \tilde x_o(t)$
remains constant in time for the linear system:
\[
\dot{\tilde x}_o = A_o\tilde x_o= 2\Theta R_o \tilde x_o.
\]
That is 
\begin{equation}
\label{Roconst}
\half \tilde x_o(t) \trp R_o \tilde x_o(t) = \half \tilde x_o(0) \trp R_o \tilde x_o(0) \quad \forall t \geq 0.
\end{equation}
However, $\tilde x_o(t) = e^{2\Theta R_ot}\tilde x_o(0)$ and $R_o > 0$. Therefore, it follows from (\ref{Roconst}) that
\[
\sqrt{\lambda_{min}(R_o)}\|e^{2\Theta R_ot}\tilde x_o(0)\| \leq \sqrt{\lambda_{max}(R_o)}\|\tilde x_o(0)\|
\]
for all $\tilde x_o(0)$ and $t \geq 0$. Hence, 
\begin{equation}
\label{exp_bound}
\|e^{2\Theta R_ot}\| \leq \sqrt{\frac{\lambda_{max}(R_o)}{\lambda_{min}(R_o)}}
\end{equation}
for all $t \geq 0$.

Now since $\Theta $ and $R_o$ are non-singular,
\[
\int_0^Te^{2\Theta R_ot}dt = \half e^{2\Theta R_oT}R_o^{-1}\Theta ^{-1} - \half R_o^{-1}\Theta ^{-1}
\]
and therefore, it follows from (\ref{exp_bound}) that
\begin{eqnarray*}
\lefteqn{\frac{1}{T} \|\int_0^Te^{2\Theta R_ot}dt\|}\nonumber \\
 &=& \frac{1}{T} \|\frac{1}{2}e^{2\Theta R_oT}R_o^{-1}\Theta ^{-1} - \frac{1}{2}R_o^{-1}\Theta ^{-1}\|\nonumber \\
&\leq& \frac{1}{2T}\|e^{2\Theta R_oT}\|\|R_o^{-1}\Theta ^{-1}\| \nonumber \\
&&+ \frac{1}{2T}\|R_o^{-1}\Theta ^{-1}\|\nonumber \\
&\leq&\frac{1}{2T}\sqrt{\frac{\lambda_{max}(R_o)}{\lambda_{min}(R_o)}}\|R_o^{-1}\Theta ^{-1}\|\nonumber \\
&&+\frac{1}{2T}\|R_o^{-1}\Theta ^{-1}\|\nonumber \\
&\rightarrow & 0 
\end{eqnarray*}
as $T \rightarrow \infty$. Hence,  
\begin{eqnarray*}
\lefteqn{\lim_{T \rightarrow \infty} \frac{1}{T}\|\int_{0}^{T} \tilde x_o(t)dt\| }\nonumber \\
&=& \lim_{T \rightarrow \infty}\frac{1}{T}\|\int_{0}^{T} e^{2\Theta R_ot}\tilde x_o(0)dt\| \nonumber \\
&\leq& \lim_{T \rightarrow \infty}\frac{1}{T} \|\int_{0}^{T} e^{2\Theta R_ot}dt\|\|\tilde x_o(0)\|\nonumber \\
&=& 0.
\end{eqnarray*}
This implies
\[
\lim_{T \rightarrow \infty} \frac{1}{T}\int_{0}^{T} \tilde x_o(t)dt = 0
\]
and hence, it follows from (\ref{tildexo}) and (\ref{Coxtilde}) that
\[
\lim_{T \rightarrow \infty} \frac{1}{T}\int_{0}^{T} z_o(t)dt = \left[\begin{array}{l}1\\1\\\vdots\\1\end{array}\right]z_p.
\]

Also, (\ref{zp_const}) implies 
\[
\lim_{T \rightarrow \infty} \frac{1}{T}\int_{0}^{T} z_p(t)dt = \left[\begin{array}{l}1\\1\\\vdots\\1\end{array}\right]z_p.
\]
Therefore, condition (\ref{average_convergence}) is satisfied. Thus, we have established the following theorem.

\begin{theorem}
\label{T1}
Consider a quantum plant of the form (\ref{plant}) where  $A_p = 0$. Then  the matrices $R_{oi}>0$, $R_{ci}$, $C_{oi}$, $i = 1,2,...N$ given as in 
(\ref{RciRoi}), (\ref{alphabeta}), (\ref{mui}), (\ref{muN})  will define a distributed direct coupled quantum observer achieving time-averaged consensus convergence for this quantum plant.
\end{theorem}

\begin{remark}
\label{R2}
The distributed quantum observer constructed above is determined by choice of the positive parameters $\tilde \mu_1, \tilde \mu_2,\ldots,\tilde \mu_N$. A number of possible choices for these parameters could be considered. One choice is to choose all of these parameters to be the same as $\tilde \mu_i = \omega_0$ for $i=1,2,\ldots,N$ where $\omega_0 > 0$ is a frequency parameter. This choice will mean that all of the oscillator frequencies in the distributed observer, except for the last one, will be the same, $\omega_i = 2\omega_0$ for $i=1,2,\ldots,N-1$ and $ \omega_N = \omega_0$. In order to have distinct oscillator frequencies in the distributed observer,  we can choose $\tilde \mu_i = i\omega_0$ for $i=1,2,\ldots,N$. This would yield $ \omega_i = (2i+1)\omega_0$ for $i=1,2,\ldots,N-1$ and $\omega_N = N\omega_0$. This means that only odd harmonics of the fundamental frequency $\omega_0$ are used. Alternatively, in order to obtain both odd and even harmonics of the fundamental frequency $\omega_0$, for the case in which $N$ is even, we can choose 
\begin{equation}
\label{mu2i}
\tilde \mu_{2i} = \tilde \mu_{2i-1} = \omega_0(\frac{N}{2}+1-i) > 0
\end{equation}
for $i=1,2,\ldots,\frac{N}{2}$. This leads to 
\begin{equation}
\label{omegai}
\omega_i = \omega_0(N+1-i)
\end{equation}
for $i=1,2,\ldots,N$. A similar choice can be derived for the case in which $N$ is odd. 

Another possible approach is to choose the parameters $\tilde \mu_1, \tilde \mu_2,\ldots,\tilde \mu_N$  randomly with a uniform distribution on $[0,\omega_0N]$. 
\end{remark}

\begin{remark}
\label{R3}
The proof of the above theorem relies an a graph theoretic argument used in the proof of Lemma \ref{L1}. This motivates a possible extension of the result in which the distributed direct coupled observer corresponding to the weighted graph in Figure \ref{F4} is replaced by a more general distributed direct coupled observer corresponding to a more general weighted graph.
\end{remark}

\section{Illustrative Example}
We now present some numerical simulations to illustrate the direct coupled distributed quantum observer described in the previous section. We choose the quantum plant to have $A_p = 0$ and $C_p = [1~~0]$. That is, the variable to be estimated by the quantum observer is the position operator of the quantum plant; i.e., $z_p(t) = q_p(t)$ where $x_p(t) = \left[\begin{array}{l}q_p(t)\\p_p(t)\end{array}\right]$. For the distributed quantum observer, we choose $N=5$ so that the distributed quantum observer has five elements. Also, as discussed in Remark \ref{R2}, we choose the parameters $\tilde \mu_1, \tilde \mu_2,\ldots,\tilde \mu_N$ so that $\tilde \mu_i = i\omega_0$ for $i=1,2,\ldots,N$ where $\omega_0 =1$. Then the corresponding distributed quantum observer is defined by equations (\ref{RciRoi}), (\ref{alphabeta}), (\ref{mui}), (\ref{muN}). 

 The augmented plant-observer system is described by the equations (\ref{augmented_system}), (\ref{xaRa}). Then, we can write
\[
x_a(t) = 
\Phi(t) x_a(0)
\]
where 
\[
\Phi(t)
= e^{A_a t}.
\]
Thus, the plant variable to be estimated $z_p(t)$ is given by
\begin{eqnarray*}
z_p(t) &=& e_1C_a\Phi(t) x_a(0) \nonumber \\
&=& \sum_{i=1}^{2N+2}e_1C_a\Phi_i(t) x_{ai}(0)
\end{eqnarray*}
where 
\[
C_a = \left[\begin{array}{ll}C_p & 0 \\0 & C_o\end{array}\right],
\]
$e_1$ is the first unit vector in the standard basis for $\rbb^{N+1}$, $\Phi_i(t)$ is the $i$th column of the matrix $\Phi(t)$ and
$x_{ai}(0)$ is the $i$th component of the vector $x_a(0)$.  We plot each of the quantities  
$e_1C_a\Phi_1(t),e_1C_a\Phi_2(t),\ldots,e_1C_a\Phi_{2N+2}(t) $ in Figure \ref{F5}(a). 
\begin{figure}%
\centering
\subfloat[][]{\includegraphics[width=3.9cm]{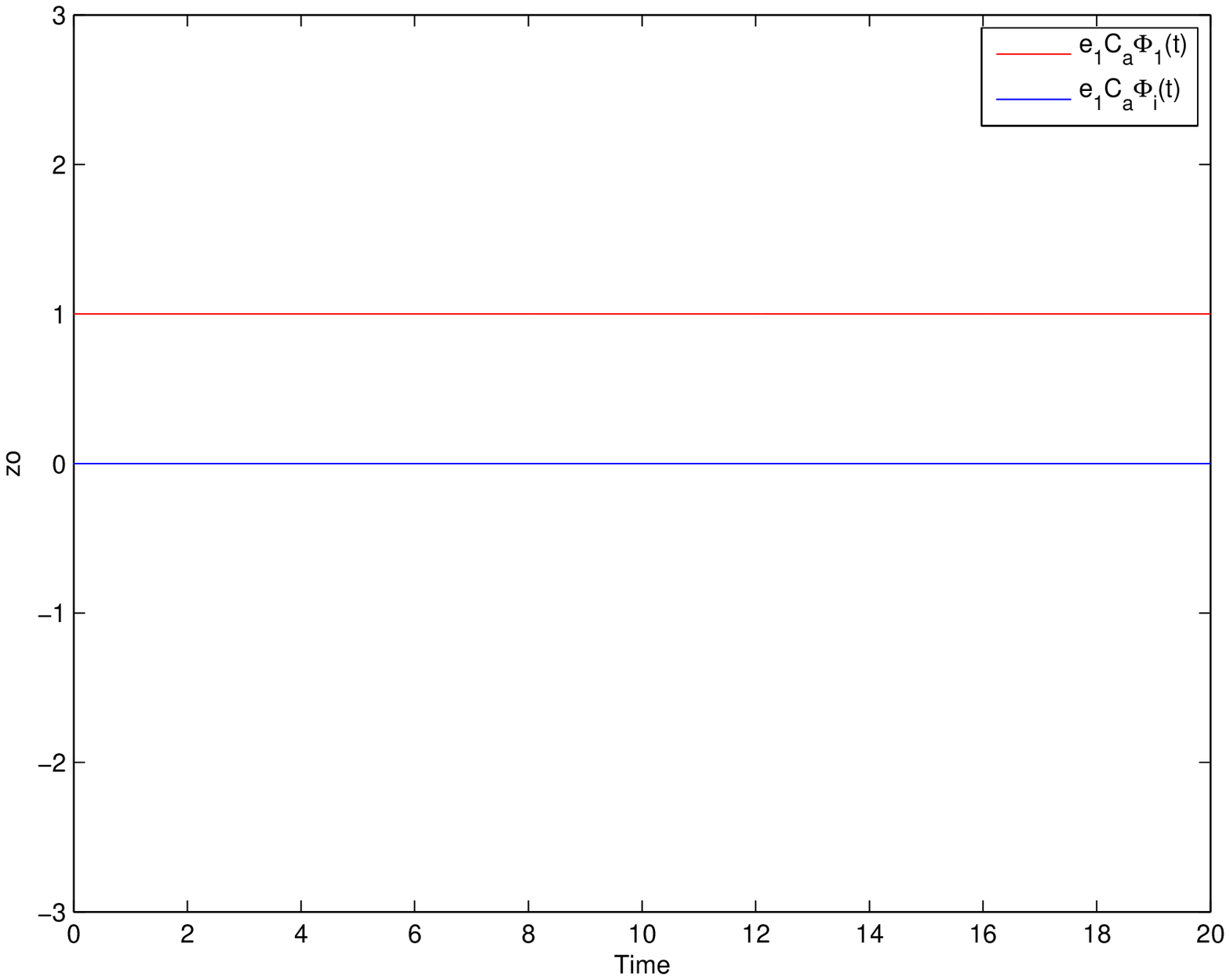}}%
\qquad
\subfloat[][]{\includegraphics[width=3.9cm]{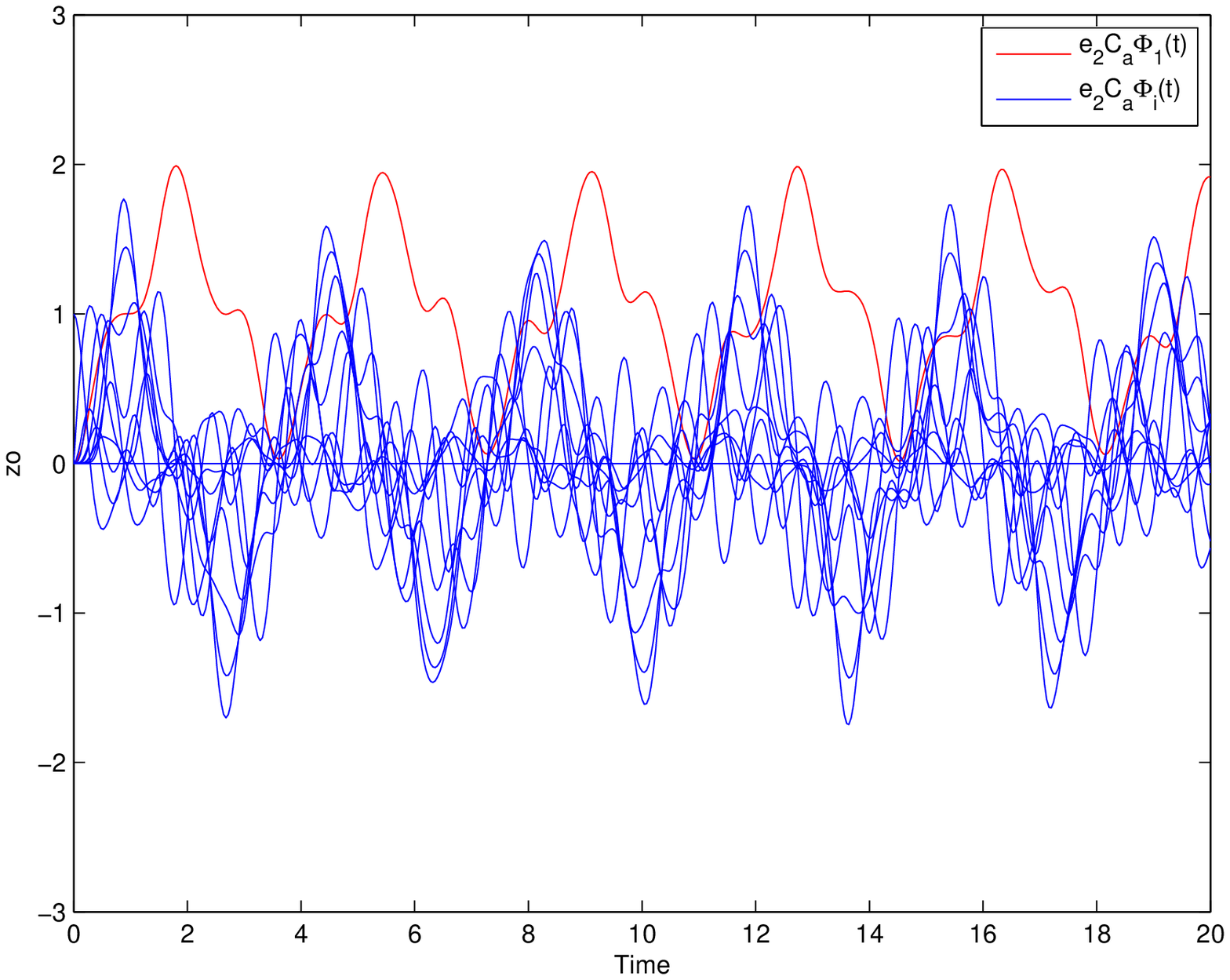}}\\
\subfloat[][]{\includegraphics[width=3.9cm]{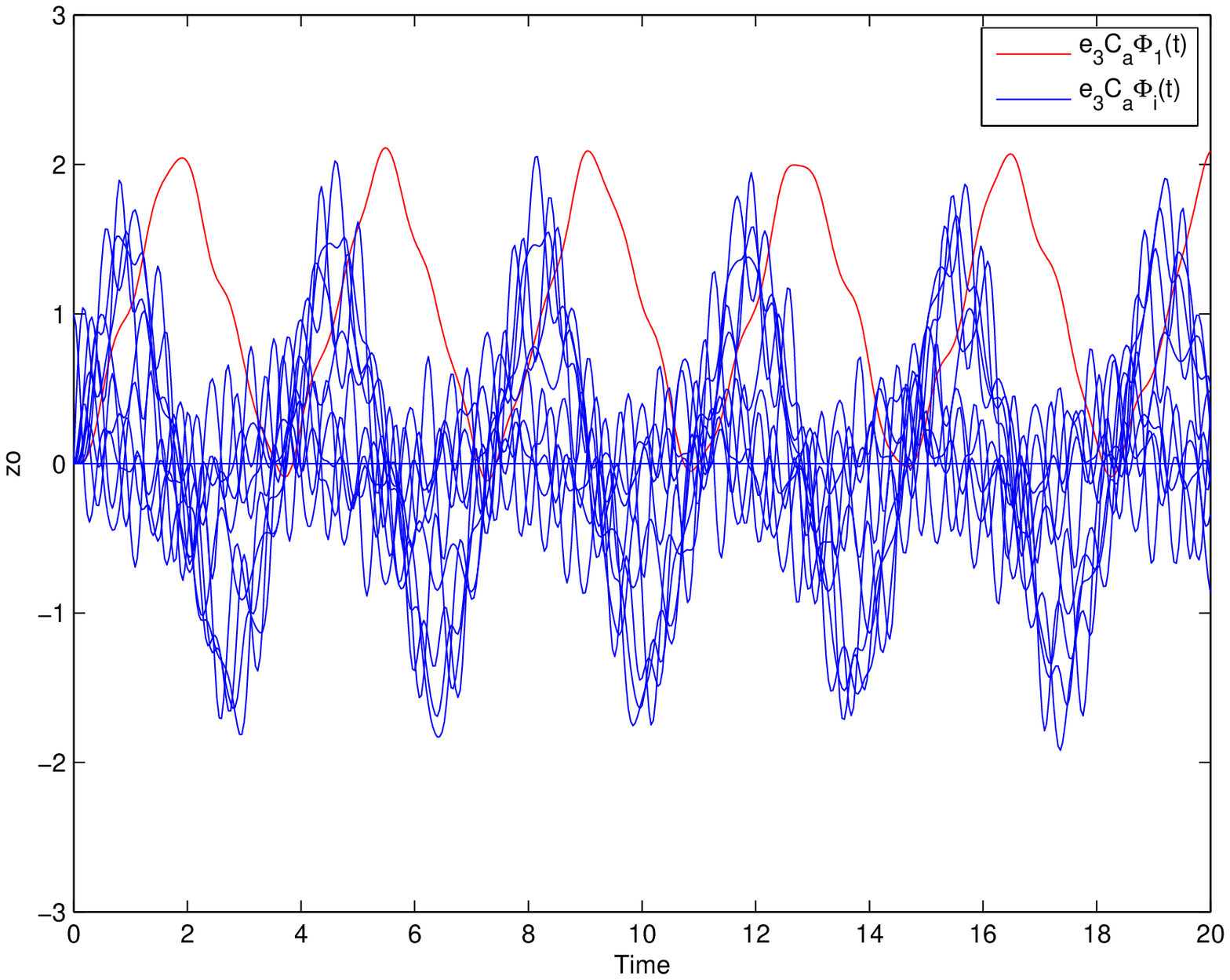}}%
\qquad
\subfloat[][]{\includegraphics[width=3.9cm]{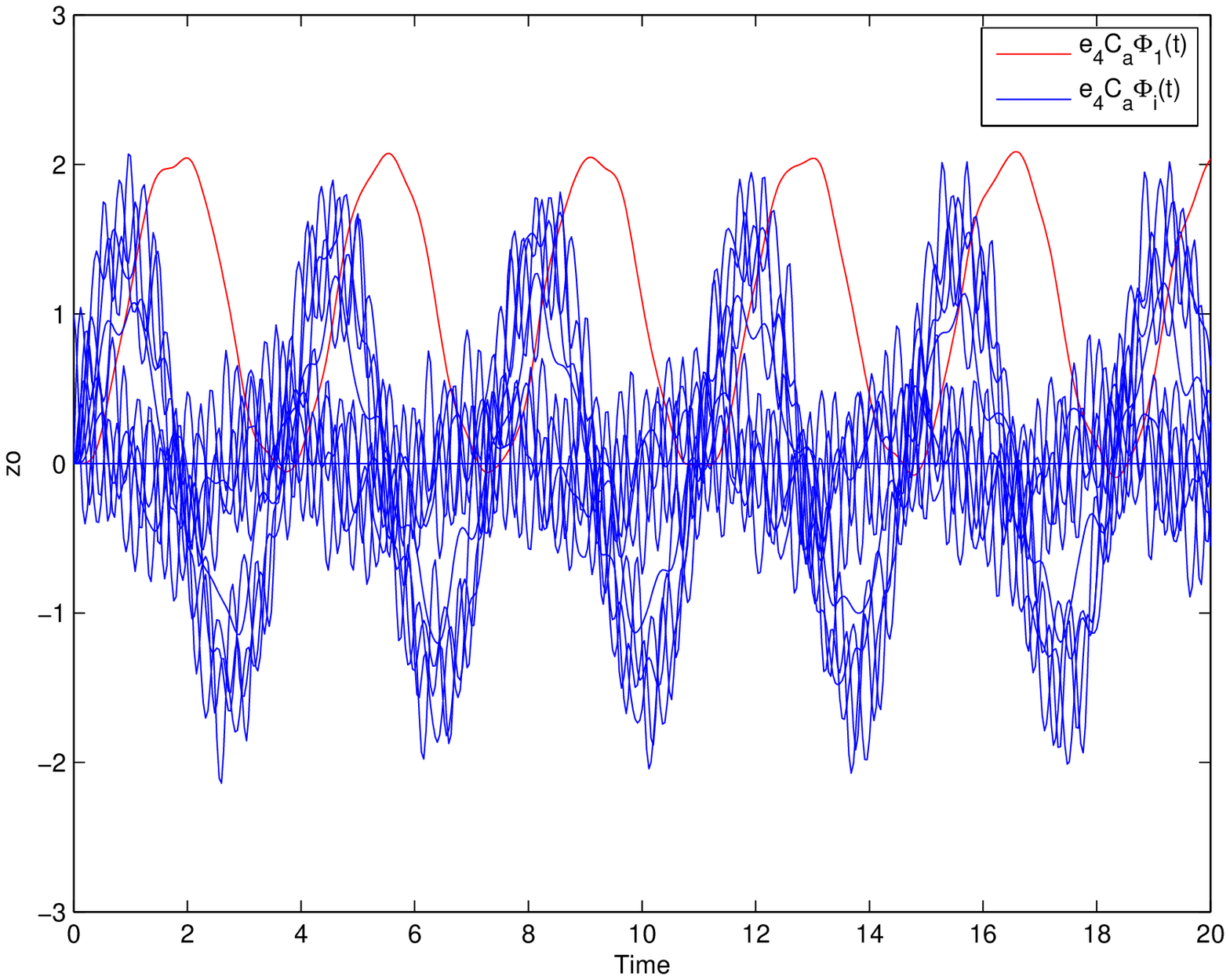}}\\
\subfloat[][]{\includegraphics[width=3.9cm]{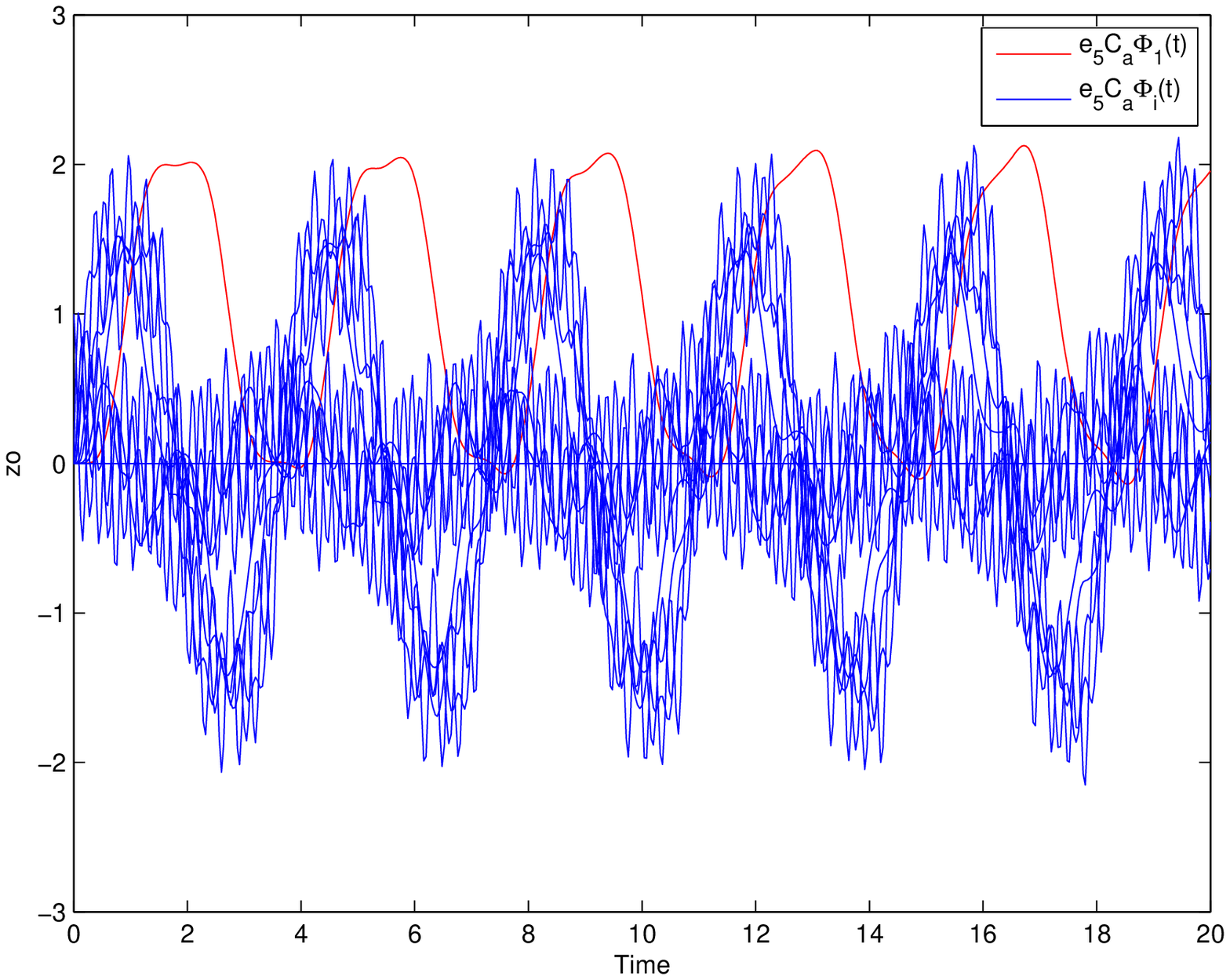}}%
\qquad
\subfloat[][]{\includegraphics[width=3.9cm]{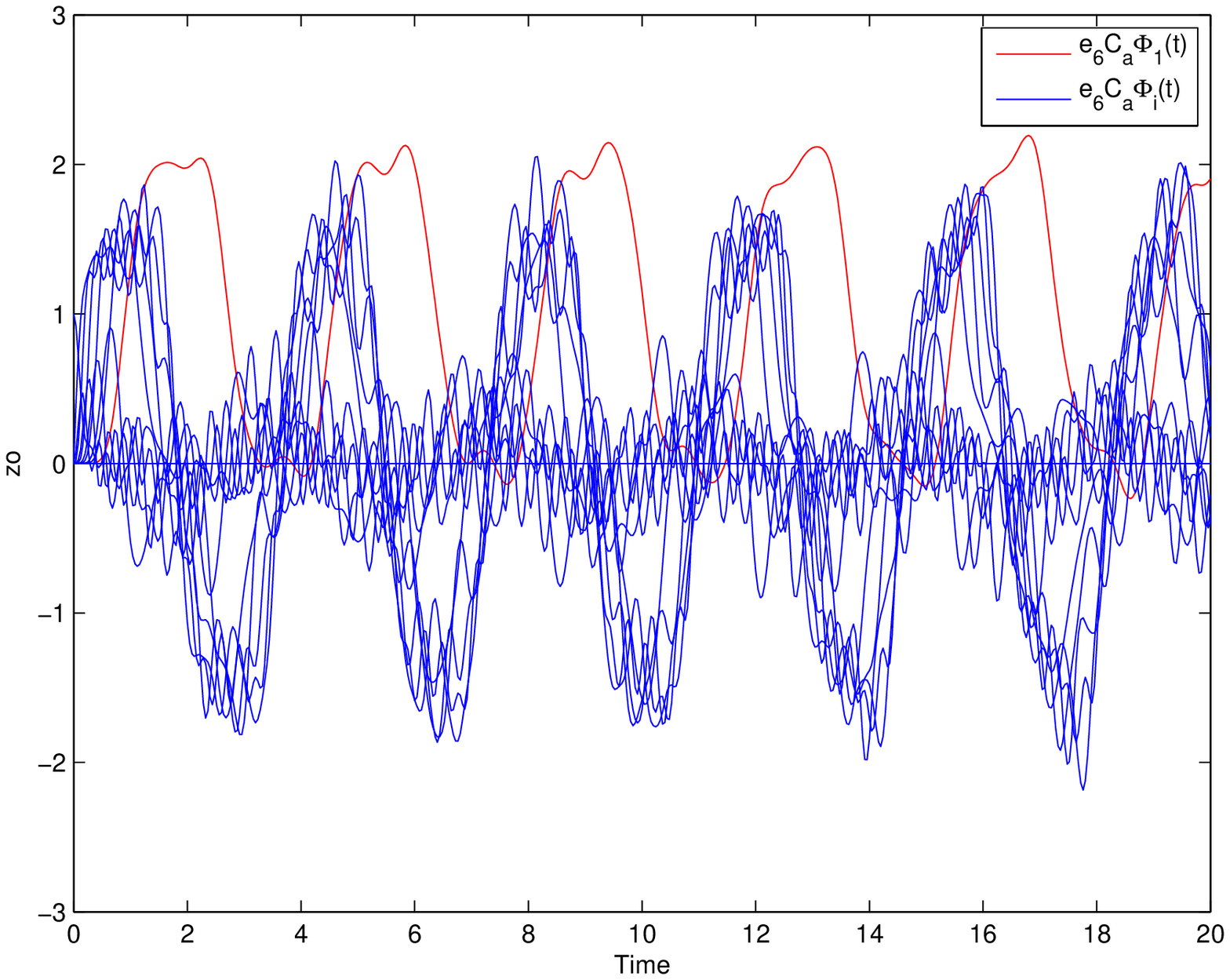}}\\
\caption{Coefficients defining (a) $z_p(t)$, (b) $z_{o1}(t)$, (c) $z_{o2}(t)$, (d) $z_{o3}(t)$, (e) $z_{o4}(t)$, and (f) $z_{o5}(t)$.}%
\label{F5}%
\end{figure}

From this figure, we can see that $e_1C_a\Phi_1(t)\equiv 1$ and  $e_1C_a\Phi_{2}(t)\equiv 0$, $e_1C_a\Phi_{2}(t)\equiv 0$, $\ldots$, $e_1C_a\Phi_{2N+2}(t)\equiv 0$, and $z_p(t)$ will remain constant at $z_p(0)$ for all $t \geq 0$.

We now consider the output variables of the distributed quantum observer $z_{oi}(t)$ for $i=1,2,\ldots,N$ which are given by
\[
z_{oi}(t) = \sum_{j=1}^{2N+2}e_{i+1}C_a\Phi_j(t) x_{aj}(0)
\]
where $e_{i+1}$ is the $(i+1)$th unit vector in the standard basis for $\rbb^{N+1}$. We plot each of the quantities  
$e_{i+1}C_a\Phi_1(t),e_{i+1}C_a\Phi_2(t),\ldots,e_{i+1}C_a\Phi_{2N+2}(t)$ in Figures  \ref{F5}(b) - \ref{F5}(f).





Also, we can consider the spatial average obtained by averaging over each of the distributed observer outputs:
\[
z_{os}(t) = \frac{1}{N}\sum_{i=1}^Nz_{oi}(t) =\frac{1}{N}\sum_{i=1}^N\sum_{j=1}^{2N+2}e_{i+1}C_a\Phi_j(t) x_{aj}(0). 
\]
Then we plot each of the quantities 
$\frac{1}{N}\sum_{i=1}^Ne_{i+1}C_a\Phi_1(t)$, $\frac{1}{N}\sum_{i=1}^Ne_{i+1}C_a\Phi_2(t)$, $\ldots$ , $\frac{1}{N}\sum_{i=1}^Ne_{i+1}C_a\Phi_{2N+2}(t)$
in Figure \ref{F6}. 

\begin{figure}[htbp]
\begin{center}
\includegraphics[width=7cm]{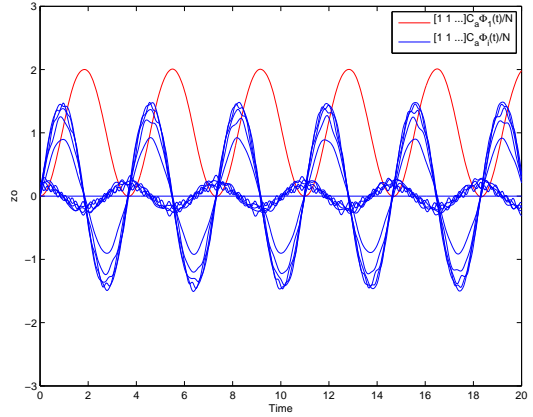}
\end{center}
\caption{Coefficients defining $z_{os}(t)$.}
\label{F6}
\end{figure}

To illustrate the time average convergence property of the quantum observer (\ref{average_convergence}), we now plot the quantities
$\frac{1}{T}\int_0^Te_{i+1}C_a\Phi_1(t)dt$, $\frac{1}{T}\int_0^Te_{i+1}C_a\Phi_2(t)dt$, $\ldots$, $\frac{1}{T}\int_0^Te_{i+1}C_a\Phi_{2N+2}(t)dt$
for $i=1,2,\ldots,N$ in Figures \ref{F7}(a)-\ref{F7}(e). These quantities determine the averaged value of the $i$th observer output
\[
z_{oi}^{ave}(T) = \frac{1}{T}\int_0^T\sum_{j=1}^{2N+2}e_{i+1}C_a\Phi_j(t) x_{aj}(0)dt
\]
for $i=1,2,\ldots,N$. 
\begin{figure}%
\centering
\subfloat[][]{\includegraphics[width=3.9cm]{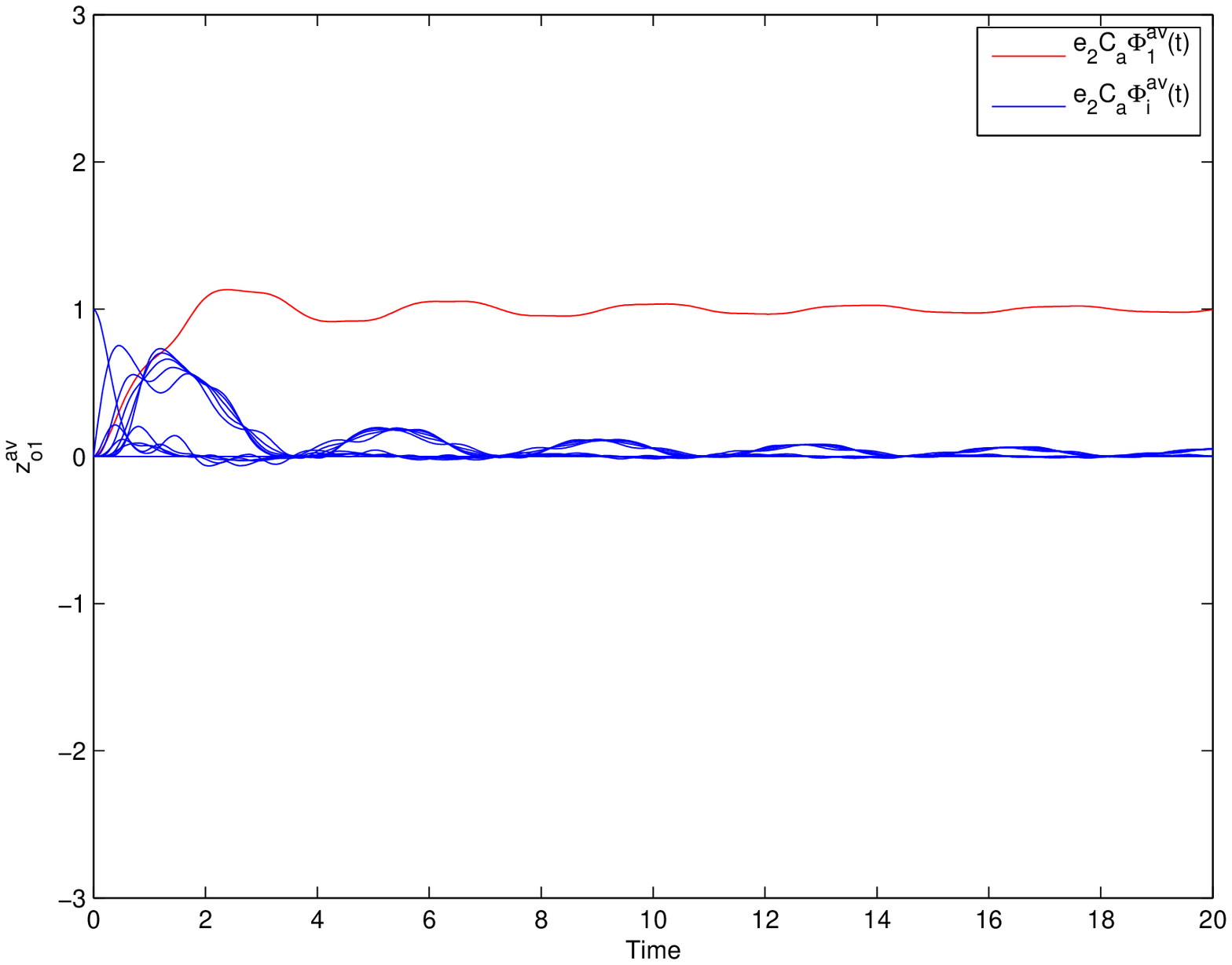}}%
\qquad
\subfloat[][]{\includegraphics[width=3.9cm]{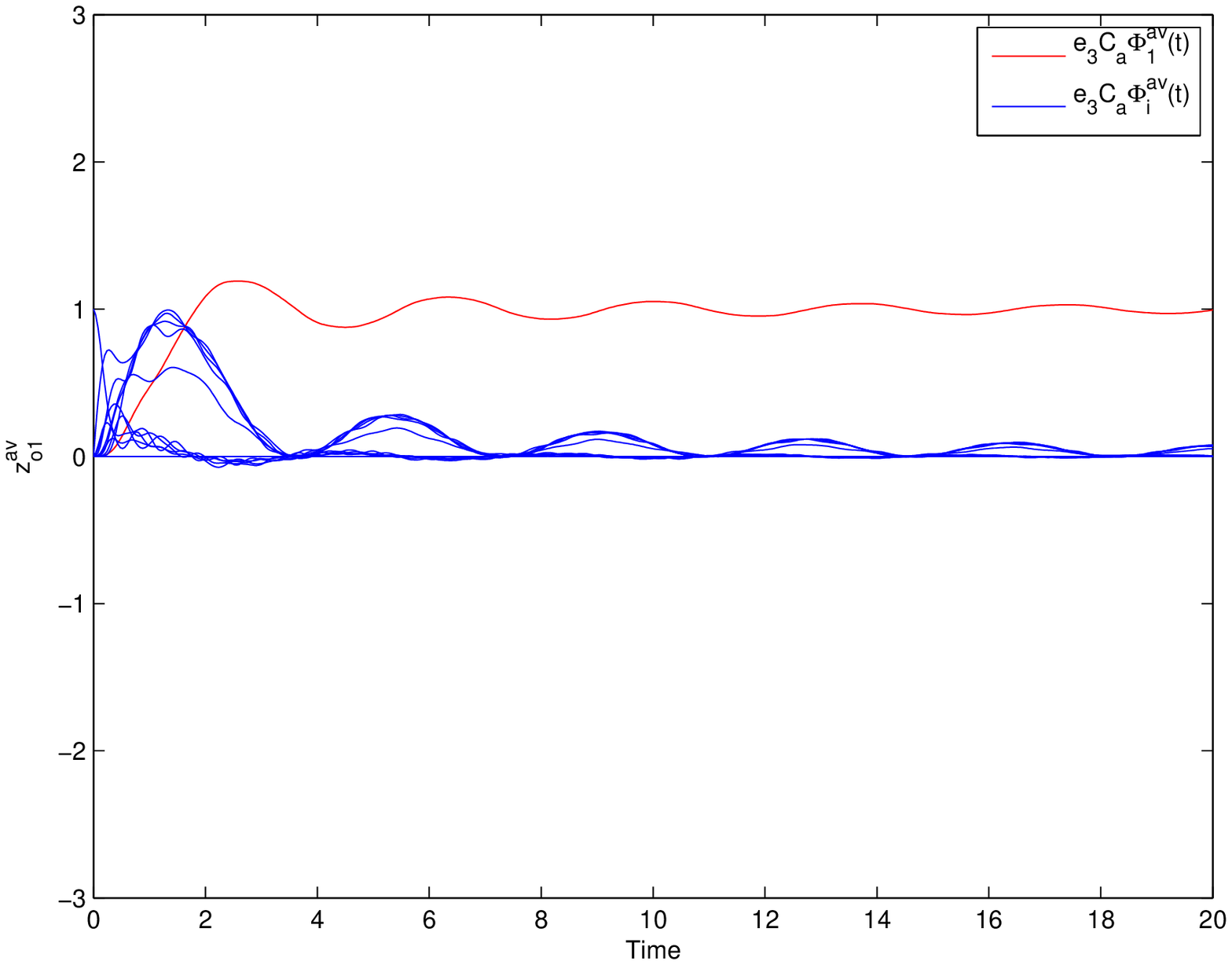}}\\
\subfloat[][]{\includegraphics[width=3.9cm]{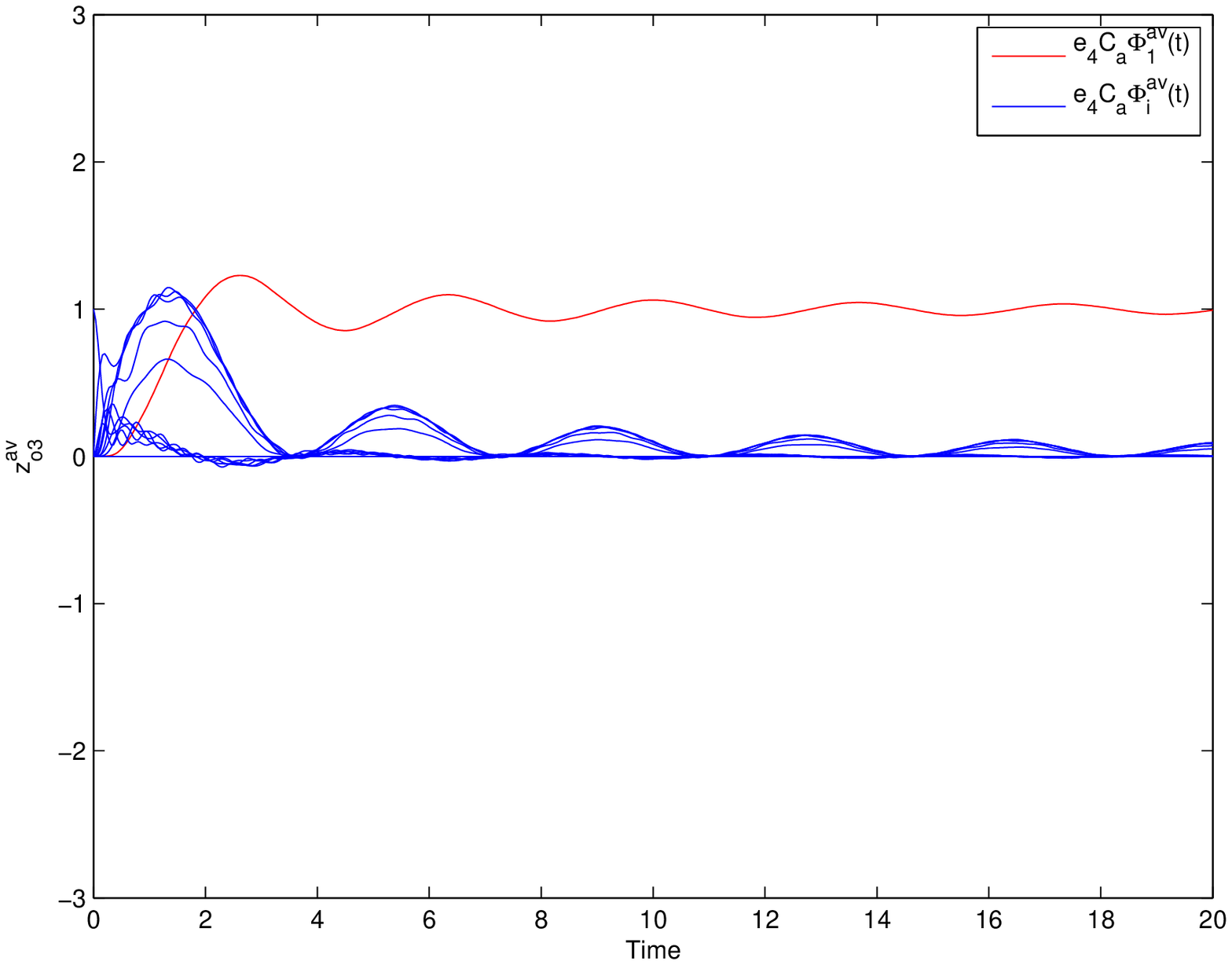}}%
\qquad
\subfloat[][]{\includegraphics[width=3.9cm]{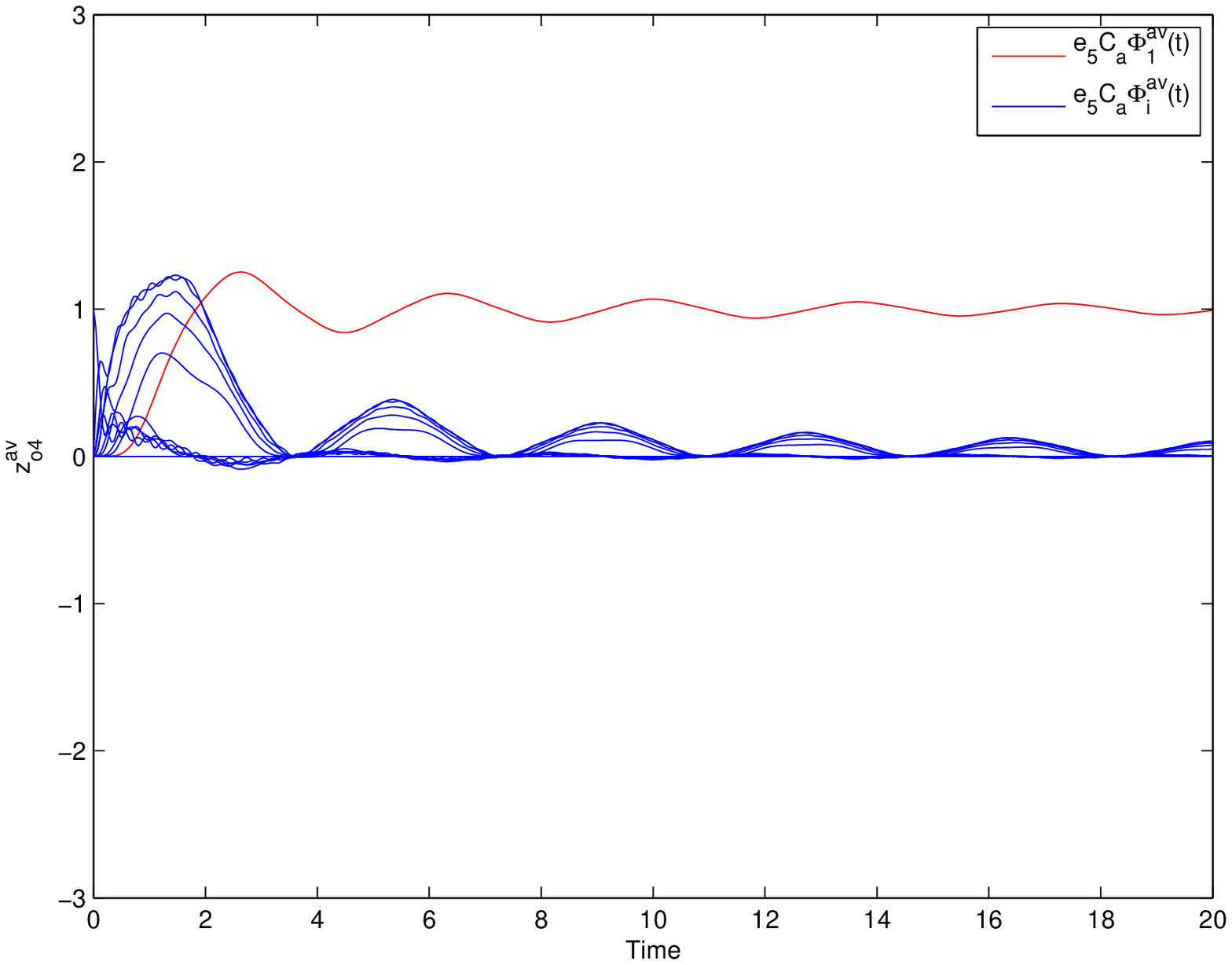}}\\
\subfloat[][]{\includegraphics[width=7cm]{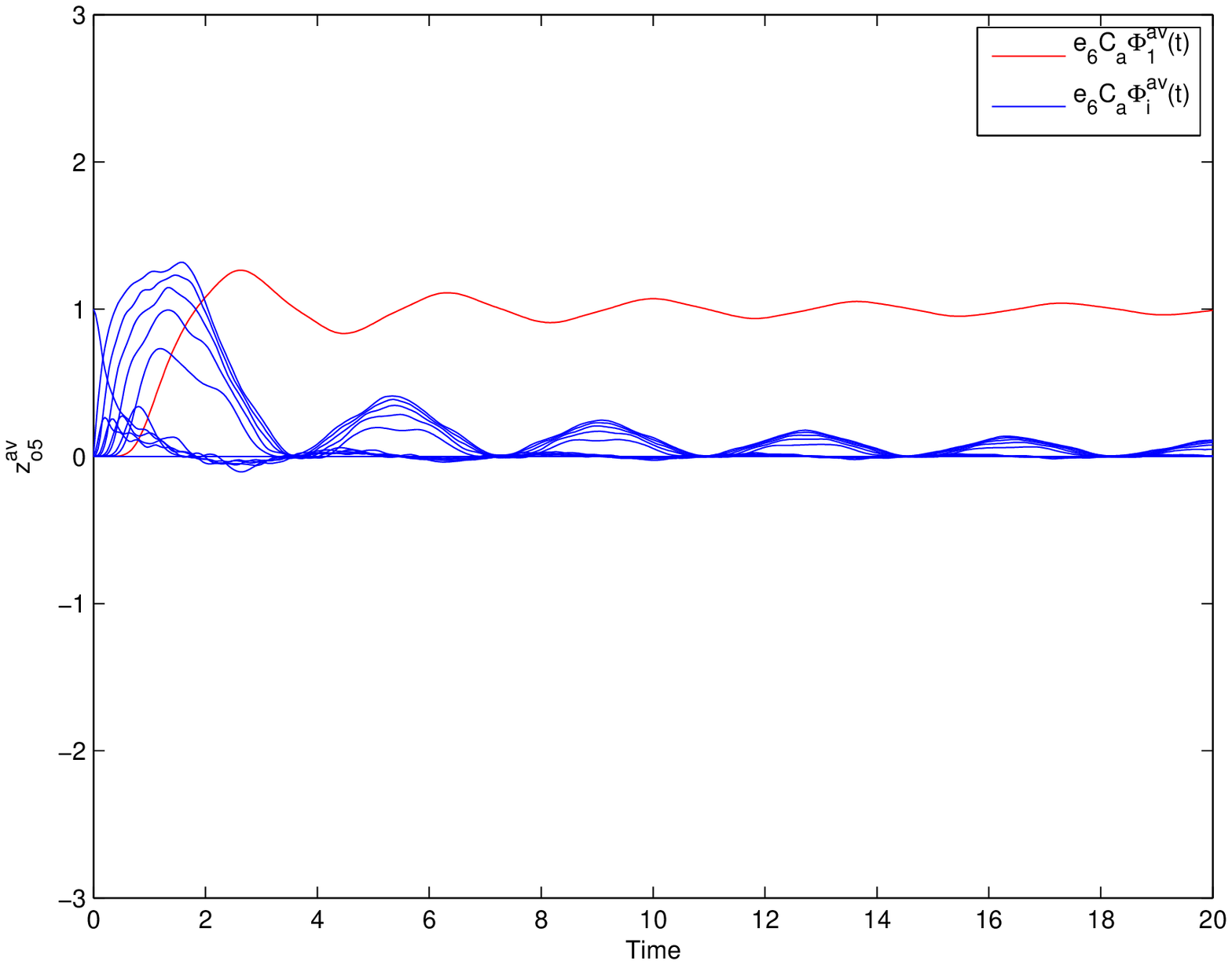}}
\caption{Coefficients defining the time average of (a)  $z_{o1}(t)$, (b) $z_{o2}(t)$, (c) $z_{o3}(t)$, (d) $z_{o4}(t)$, and (e) $z_{o5}(t)$.}%
\label{F7}%
 \end{figure}






From these figures, we can see that for each $i=1,2,\ldots,N$, the time average of $z_{oi}(t)$  converges to $z_p(0)$ as $t \rightarrow \infty$. That is, the distributed quantum observer reaches a time averaged consensus corresponding to the output of the quantum plant which is to be estimated. 

\section{Conclusions}
In this paper we have considered the construction of a distributed direct coupling observer for a closed quantum linear system in order to achieve a time averaged consensus convergence.  We have also presented an illustrative example along with simulations to investigate the consensus behavior of the distributed direct coupling observer.

\end{document}